\documentclass[10pt]{amsart}

\usepackage{amsfonts}
\usepackage{amssymb}
\usepackage{calrsfs}
\usepackage[mathscr]{eucal}
\usepackage{amsmath}
\usepackage{amsthm}
\usepackage[all]{xy}
\usepackage[mathscr]{eucal}
\usepackage{longtable}

\usepackage{hyperref}
\hypersetup{
    colorlinks=true,
    linkcolor=blue,
    citecolor=blue,
    filecolor=blue,
    urlcolor=blue
}

\newtheorem{defi}{Definition}[section]
\newtheorem{ejem}[defi]{Example}
\newtheorem{teo}[defi]{Theorem}
\newtheorem{coro}[defi]{Corollary}
\newtheorem{prop}[defi]{Proposition}
\newtheorem{lem}[defi]{Lemma}
\newtheorem{obs}[defi]{Remark}

\begin{document}

\title{Some properties of skew codes over finite fields}

\author{Luis Felipe Tapia Cuiti\~no and Andrea Luigi Tironi}

\date{\today}

\address{
Departamento de Matem\'atica,
Universidad de Concepci\'on,
Casilla 160-C,
Concepci\'on, Chile}
\email{ltapiac@udec.cl,\ atironi@udec.cl}


\thanks{The second author was partially supported 
by Proyecto VRID N. 214.013.039-1.OIN}

\maketitle

\begin{abstract}
After recalling the definition of codes as modules over skew
polynomial rings, whose multiplication is defined by using an automorphism
and a derivation, and some basic facts about them, in the first part of this paper we study some of their main algebraic and 
geometric properties. Finally, for module skew codes constructed only with an
automorphism, we give some BCH type lower bounds for their minimum distance.
\end{abstract}

\section*{Introduction}

\noindent In the framework of linear codes, the introduction of the cyclic codes has been important due to the fact that for the first time some special linear 
codes could be treated by polynomials rings via a vector space isomorphism.  
The use of Ore polynomial rings in the non-commutative setting emerged only recently in Coding Theory as source of generalizations of the cyclic codes.
The skew polynomial rings have found applications in the construction of many algebraic codes of good parameters with respect to the commutative case.
In particular, the research on codes in this setting has resulted in the discovery of many new codes with better Hamming distance than
any previously known linear code with same parameters (see for instance Tables 1, 2 and 3 of \cite{BU2}).

Inspired by the recent works \cite{BU2} and \cite{BL}, in $\S1$ we give a background material and the notion of skew generalized cyclic (GC) codes, that is, linear codes invariant by a pseudo-linear transformation (see Definition (\ref{def})). In $\S2$, we introduce some basic properties of skew GC codes and we give a result (Theorem \ref{prop}) about duals of skew GC codes that improves Theorem 23 in \cite{TT}. Moreover, in Theorem \ref{ker skew} we show a fundamental geometric property of these codes which becomes a useful tool to find codes through the
factorization of polynomials. In the commutative case ($\theta =id$), the same result delivers more geometric consequences described in Corollary \ref{ker}.
Furthermore, assumption $(\#)$ and its properties allow us to extend to the non-commutative case some of the main results of \cite{RZ} (see Propositions \ref{inv sub}
and \ref{inv sub2}). In $\S 3$ we consider the cases with trivial derivation ($\delta_{\beta}^{\theta}=0$) by studying the minimal polynomial of a semi-linear transformation
in Theorem \ref{matrixpol3} and some BCH lower bounds which generalize known results of \cite{BU2} and \cite{HT} in the non-commutative case (see Theorems \ref{bound1}, \ref{bound2},
\ref{bound3} and Corollary \ref{mds}). Finally, in $\S 4$ we give some Magma programs and examples as immediate applications of some previous results.

\section{Notation and background material}

Denote by $\theta : \mathbb{F}_q\to\mathbb{F}_q$ an automorphism 
of the finite field $\mathbb{F}_q$. Let us recall here
that if $q=p^s$ for some prime number $p$, then the map
$\tilde{\theta}:\mathbb{F}_q\to\mathbb{F}_q$ defined by
$\tilde{\theta}(a)=a^p$ is an automorphism on the field
$\mathbb{F}_q$ which fixes the subfield with $p$ elements. This
automorphism $\tilde{\theta}$ is called the {\em Frobenius
automorphism} and it has order $s$. Moreover, it is known
that the cyclic group it generates is the full group of
automorphisms of $\mathbb{F}_q$, i.e.
$\mathrm{Aut}(\mathbb{F}_q)=<\tilde{\theta}>$. Therefore, any
$\theta\in\mathrm{Aut}(\mathbb{F}_q)$ is defined as
$\theta (a):=\tilde{\theta}^t(a)=a^{p^t},$ where $a\in\mathbb{F}_q$ and $t$ is an
integer such that $0\leq t\leq s$. Furthermore, when $\theta$ will be the identity automorphism
$id: \mathbb{F}_q\to\mathbb{F}_q$, we will write simply $\theta=id$.
Finally, a $\theta$-derivation must be of the form
$\beta(\theta(a)-a)$ for any $a\in\mathbb{F}_q$ and some $\beta\in\mathbb{F}_q$.   

A pseudo-linear map (or a pseudo-linear transformation) $T:\mathbb{F}_q^n\to\mathbb{F}_q^n$ is an additive map defined by
\begin{equation}\label{def}
T(\vec{v}):=(\vec{v})\Theta\circ M+(\vec{v})\delta_{\beta}^{\theta}, \tag{$\triangle$} 
\end{equation}
where $(v_1,...,v_n)\Theta
:=(\theta(v_1),...,\theta(v_n))$, $M$ is an $n\times n$ matrix
with coordinates in $\mathbb{F}_q$ and $(v_1,...,v_n)\delta_{\beta}^{\theta}
:=(\beta(\theta(v_1)-v_1),...,\beta(\theta(v_n)-v_n))$. When $\delta_{\beta}^{\theta}=0$, we call $T$ a
\textit{semi-linear map}, or a
\textit{semi-linear transformation}.

Consider the ring structure defined on the following set:
$$R:=\mathbb{F}_q[X;\theta,\delta_{\beta}^{\theta}]=\left\{ a_sX^s+...+a_1X+a_0\ | \ a_i\in\mathbb{F}_q\ \mathrm{and}\ s\in\mathbb{N}\right\}.$$
The addition in $R$ is defined to be the usual addition of
polynomials and the multiplication is defined by the basic rule
$X\cdot a = \theta(a)X+\beta(\theta(a)-a)$ for any $a\in\mathbb{F}_q$ and extended to
all elements of $R$ by associativity and distributivity. The ring
$R$ is known as skew polynomial ring and its elements are skew
polynomials. Moreover, it is a left and right Euclidean ring whose
left and right ideals are principals.

From now on, fix a polynomial
\begin{equation}
f:=X^n-f_{n-1}X^{n-1}-\ldots - f_1X-f_0\in R \tag{*}
\end{equation}
and denote by $\pi _f:\mathbb{F}_q^n \rightarrow R/Rf$ the linear transformations which sends 
a vector $(c_0,\ldots ,c_{n-1})\in \mathbb{F}_q^n$ to the polynomial class $\left[ \sum _{i=0}^{n-1}c_iX^{i}\right]\in R/Rf$, i.e.
$$\pi _f((c_0,...,c_{n-1})):=\displaystyle
\left[\sum _{i=0}^{n-1}c_iX^i\right] \in R/Rf\qquad \forall \
(c_0,...,c_{n-1})\in\mathbb{F}_q^n.$$

\noindent With an abuse of notation, a polynomial $c\in R$ and its class $[c]\in R/Rf$ will be denoted sometimes with the same letter $c$
for simplicity.

\smallskip

Let us introduce here a generalization of the skew cyclic codes. 

\begin{defi}\label{skew GCC}
A linear code $\mathcal{C} \subseteq \mathbb{F}_q^n$ is called a $(M,\Theta,\delta_{\beta})$-\textbf{code} if
$T\star \mathcal{C} \subseteq \mathcal{C}$, where $T$ is as in {\em (\ref{def})} and $T\star \mathcal{C}:=\{T(\vec{v})\ |\ \vec{v}\in\mathcal{C} \}$.
Moreover, $\mathcal{C} \subseteq \mathbb{F}_q^n$ is called a $(f,\Theta,\delta_{\beta}^{\theta})$-\textbf{skew generalized cyclic (GC) code} (or, for simplicity, a \textbf{skew GC code}) 
if $T_f \star \mathcal{C} \subseteq \mathcal{C}$, where $T_f$ is the pseudo-linear map defined by  
$T_f(\vec{v}):=(\vec{v})\Theta\circ A+(\vec{v})\delta_{\beta}^{\theta}$ with
\begin{equation}
A:=\left(
\begin{array}{c|ccc}
0&1&&\\
\vdots &&\ddots&\\
0&&&1\\
\hline
f_{0}& f_{1}&\cdots &f_{n-1}
\end{array}
\right) . \tag{**}
\end{equation}

Finally, when $\theta = id$, a $(f,\Theta,\delta_{\beta}^{\theta})$-skew generalized cyclic (GC) code is called an $f$-\textbf{generalized cyclic (GC) code},
or simply a \textbf{GC code}.
\end{defi}

\begin{obs}
Suppose that $\theta =id$ and $f = X^n-f_{n-1}X^{n-1}\ldots - f_1X-f_0$ is the minimal polynomial of an $n\times n$ matrix $M$. 
Then by \cite[Lemma 6.7.1]{H} there is a basis $\{ \vec{r}_{1},\ldots ,\vec{r}_{n}\}$ of $\mathbb{F}_q^n$ such that 
$$M=S\left(\begin{array}{c|ccc}
0&1&&\\
\vdots&&\ddots&\\
0&&&1\\
\hline
f_{0}&f_{1}&\cdots&f_{n-1}
\end{array}
\right)S^{-1}$$
where $S=\left( \begin{array}{c|c|c}
\vec{r}_{1t}&\cdots &\vec{r}_{nt}
\end{array}
\right) .$
Therefore the $n\times n$ matrix $M$ is similar to its rational canonical form $A$, i.e. there exists a non-singular matrix $S$ such that 
\begin{equation}
M=SAS^{-1}. \tag{$\circ$}
\end{equation}
Let $\mathcal{C}_{M} \subseteq \mathbb{F}_q^n$ be a linear code invariant by the matrix $M$. 
Define $$\mathcal{C}_A := \mathcal{C}_M\star S:=\{ \vec{c}S:\vec{c}\in \mathcal{C}_M\}\ .$$ Then by $(\circ)$ we obtain
$$\mathcal{C}_A \star A=\mathcal{C}_M\star (SA)=\mathcal{C}_M\star (S(S^{-1}MS))=(\mathcal{C}_M\star M)\star S\subseteq \mathcal{C}_M\star S=\mathcal{C}_A,$$
i.e. $\mathcal{C}_A$ is invariant by $A$. Since $S$ is an invertible matrix, this shows that we can construct a \textit{one-to-one} correspondence between the set of linear codes invariant by $A$ and the set of linear codes invariant by $M$. 
\end{obs}

\section{Basic properties of skew GC codes}

In this section, we give some algebraic and geometric properties of skew generalized cyclic (GC) codes.

\begin{obs}\label{rem 1}
Let $T$ be any pseudo-linear transformation on $\mathbb{F}_q^n$. If $p=p_0+p_1X+ \dots +p_mX^m$ is a polynomial, then $p(T)=p_0+p_1T+ \dots +p_mT^m$ is not 
in general a pseudo-linear transformation. On the other hand, a linear subspace $U\subseteq\mathbb{F}^n_q$ is $T$-invariant, i.e. $T\star U\subseteq U$, if and only if
$p(T) \star U\subseteq U$ for any polynomial $p$.
\end{obs}

\begin{obs}\label{rem 1bis}
We have $p(X)\cdot \pi _f (\vec{v})=\pi _f(p(T_f)(\vec{v}))$ for any polynomial $p$ and all $\vec{v}\in\mathbb{F}_q^n$.
\end{obs}

\noindent From Remarks \ref{rem 1} and \ref{rem 1bis}, we deduce the following characterization of 
a $(f,\Theta,\delta_{\beta}^{\theta})$-skew GC code. 

\begin{prop}\label{corr. ideal2}
Let $\mathcal{C}$ be a non-empty subset of $\mathbb{F}_q^n$. Then

\medskip

\ $\mathcal{C}$ is a $(f,\Theta,\delta_{\beta}^{\theta})$-skew GC code $\iff$ $\pi _f(\mathcal{C})$ is a principal left ideal of $R/Rf$.
\end{prop}

\begin{obs}
By {\em Proposition \ref{corr. ideal2}}, we see that to produce examples of $(f,\Theta,\delta_{\beta}^{\theta})$-skew GC codes one can focus on finding
right divisors of $f$. 
\end{obs}

\begin{defi}
For any $(f,\Theta,\delta_{\beta}^{\theta})$-skew GC code $\mathcal{C} \subseteq
\mathbb{F}_q^n$, the generator polynomial of $\pi _f(\mathcal{C})\subseteq R/Rf$ is called the
\textbf{generator polynomial} of $\mathcal{C}$.

Therefore, we will write $\mathcal{C} =(g)_{n,q}^{\theta,\delta_{\beta}^{\theta}}$ for
a $(f,\Theta,\delta_{\beta}^{\theta})$-skew GC code $\mathcal{C}\subseteq\mathbb{F}_q^n$ with
generator polynomial $g=g_{0}+g_{1}X+\cdots +g_{n-k}X^{n-k}\in
R/Rf$.
\end{defi}

From \cite{BU2} we know that a generator matrix of a $(f,\Theta,\delta_{\beta}^{\theta})$-skew GC code 
$\mathcal{C}=(g)_{n,q}^{\theta,\delta_{\beta}^{\theta}}$ is given by
$$G:=\left(
\begin{array}{c}
\vec{g}  \\
T_f(\vec{g}) \\
\vdots  \\
T_f^{k-1}(\vec{g})
\end{array}\right) ,$$
where $\vec{g}=\pi_f^{-1}(g)$.

\smallskip

About the dual code $\mathcal{C}^{\perp }$ of a $(M,\Theta,\delta_{\beta}^{\theta})$-code $\mathcal{C}\subseteq \mathbb{F}_q^n$,
we can give the following

\begin{teo}\label{prop} 
If $\mathcal{C}\subseteq \mathbb{F}_q^n$ is a $(M,\Theta,\delta_{\beta}^{\theta})$-code, then
$\mathcal{C}^{\perp }$ is a $((M_t)_{\theta^{-1}},\Theta^{-1},\delta_{\theta^{-1}(\beta)}^{\theta^{-1}})$-code,
where $W_{\theta^{-1}}:=[\theta^{-1} (a_{ij})]$ and $W_t$ is the transpose matrix of a matrix $W=[(a_{ij})]$. 
\end{teo}

\begin{proof}
Denote by $T'$ the pseudo-linear map defined by $T'(\vec{v}):=(\vec{v})\Theta^{-1}\circ (M_t)_{\theta^{-1}}+(\vec{v})\delta_{\theta^{-1}(\beta)}^{\theta^{-1}}$.
For any $\vec{c}\in\mathcal{C}$ and $\vec{a}\in\mathcal{C}^{\perp}$, note that 
$$0=\vec{a}\cdot T_f(\vec{c})=\vec{a}\cdot ((\vec{c})\Theta\circ M)+\vec{a}\cdot((\vec{c})\delta_{\beta}^{\theta})=$$
$$=(\vec{a}M_t)\cdot ((\vec{c})\Theta)+\vec{a}\cdot((\vec{c})\Theta\beta)=(\vec{a}M_t+\beta\vec{a})\cdot((\vec{c})\Theta),$$
i.e. $(\vec{a}M_t+\beta\vec{a})\cdot((\vec{c})\Theta)=0$.
Thus by \cite[Lemma 4]{TT} we conclude that
$$(T'(\vec{a}))\cdot \vec{c}=\left((\vec{a})\Theta^{-1}\circ (M_t)_{\theta^{-1}}+(\vec{a})\delta_{\theta^{-1}(\beta)}^{\theta^{-1}}\right)\cdot\vec{c}=$$
$$=\left((\vec{a})\Theta^{-1}\circ (M_t)_{\theta^{-1}}+((\vec{a})\Theta^{-1}-\vec{a})\theta^{-1}(\beta)\right)\cdot\vec{c}$$ 
$$=\left((\vec{a})(M_t)\circ\Theta^{-1}+(\beta\vec{a})\Theta^{-1}\right)\cdot\vec{c}=\left([(\vec{a})(M_t)+(\beta\vec{a})]\Theta^{-1}\right)\cdot\vec{c}=0.$$ 
for every $\vec{a}\in\mathcal{C}^{\perp}$ and $\vec{c}\in\mathcal{C}$.
\end{proof}

\begin{defi}
A polynomial $p\in R$ is invariant if $Rp=pR$. Moreover, the set of all invariant polynomials in $R$ will be denoted by $N(R)$.
\end{defi}

Let us prove now the following

\begin{teo}\label{ker skew}
Let $f\in R$ be as in {\em (*)} and let $T_f$ be its associated pseudo-linear transformation as in {\em (**)}.
Then the following properties hold:
\begin{enumerate}
\item If $f=h\cdot g$ for some $h,g \in R$, then
$$Rg/Rf=\pi_f(\mathrm{Ker}\ h(T_f))\iff h\in N(R);$$
\item If $f=h\cdot g=g\cdot h'$ for some $h,h',g \in R$, e.g. when $g\in N(R)$, then $\deg g$ linearly independent columns of 
the $n\times n$ matrix
$$\left(\begin{array}{c}
\pi_f^{-1}(h') \\
T_f(\pi_f^{-1}(h')) \\
\vdots \\ 
T_f^{n-1}(\pi_f^{-1}(h'))
\end{array}\right)$$
form a basis of $\mathcal{C}^{\perp}$, where $\mathcal{C}=(g)_{n,q}^{\theta,\delta_{\beta}^{\theta}}$. 

\end{enumerate}
\end{teo}

\begin{proof}
(1) Assume that $h\in N(R)$. If $a\in Rg/Rf$ then $a=\alpha g$ and
$$h\alpha g = \alpha' hg = \alpha' f = 0\in R/Rf.$$ Then we have
$$\vec{0}=\pi_f^{-1}(h\alpha g)=h(T_f)\pi_f^{-1}(\alpha g).$$ This shows that $a=\alpha g\in \pi_f(\mathrm{Ker}\ h(T_f)),$
that is, $Rg/Rf\subseteq\pi_f(\mathrm{Ker}\ h(T_f)).$

\medskip

\noindent On the other hand, if $v\in\pi_f(\mathrm{Ker}\ h(T_f))$ then $h(T_f)\pi_f^{-1}(v)=\vec{0}$. 
Hence $hv=\beta f$ for some $\beta\in R$ and this gives
$$hv=\beta f =\beta hg=h\beta'g,$$
i.e. $v=\beta' g\in Rg/Rf$. Thus $\pi_f(\mathrm{Ker}\ h(T_f))\subseteq Rg/Rf$, i.e. $Rg/Rf=\pi_f(\mathrm{Ker}\ h(T_f)).$

Conversely, suppose that $Rg/Rf=\pi_f(\mathrm{Ker}\ h(T_f)).$ Let $\alpha\in R$ and write $\alpha g=\pi_f(\vec{v})$, where $\vec{v}\in\mathrm{Ker}\ h(T_f)$.
Then
$$h\alpha g=h\pi_f(\vec{v})=\pi_f(h(T_f)\vec{v})=\pi_f(\vec{0})=0\in R/Rf.$$
This shows that there exists an element $q\in R$ such that $h\alpha g=qf=qhg$, i.e. $h\alpha=qh$. Hence $h\in N(R)$.

\medskip

\noindent (2) Let $\vec{c}\in\mathcal{C}$ and write $\pi_f(\vec{h'})=h'$ and $\pi_f(\vec{c})=c=\alpha g$, for some $\alpha,c\in R$.
Then we get 
$$\pi_f(c(T_f)(\vec{h'}))=ch'=(\alpha g)h'=\alpha (gh')=\alpha f=0\in R/Rf,$$
i.e. $c(T_f)(\vec{h'})=\vec{0}$. Write $\vec{c}=(c_0,\dots,c_{n-1})$. Therefore we have
$$\vec{0}=c(T_f)(\vec{h'})=c_0\vec{h'}+c_1T_f(\vec{h'})+...+c_{n-1}T_f^{n-1}(\vec{h'}).$$
This shows that 
$$(c_0,\dots,c_{n-1})\cdot \left(\begin{array}{c}
\vec{h'} \\
T_f(\vec{h'}) \\
\vdots \\ 
T_f^{n-1}(\vec{h'})
\end{array}\right)=\vec{0}$$
for any $\vec{c}=(c_0,\dots,c_{n-1})\in\mathcal{C}$. Finally, since $\pi_f(T_f^{k}(\vec{h'}))=X^kh'$ for $k=0,...,n-\deg h'-1=\deg g -1$, we see that $\{\vec{h'}, T_f(\vec{h'}), \dots , T_f^{\deg g-1}(\vec{h'}) \}$ are linearly independent.
\end{proof}

From the above result, we can deduce the following consequences.

\begin{coro}\label{cor}
Suppose that the same hypothesis as in {\em Theorem \ref{ker skew}}{\em (1)} holds. Then

\smallskip

\quad $\mathrm{Ker}\ h(T_f)\subseteq \mathbb{F}_q^n$ is a $(f,\Theta,\delta_{\beta}^{\theta})$-skew GC code $\iff\ h\in N(R).$
\end{coro}

\begin{coro}\label{ker}
Assume that $\theta=id$ and let $f\in \mathbb{F}_q[X]$. Then we have the following properties:
\begin{enumerate}
\item If $f=h\cdot g$ for some $h,g\in\mathbb{F}_q[X]$, then $(g)=\pi_f (\ker h(A))$;
\item $\ker h'(A)\subseteq \mathbb{F}_q^n$ is a vector subspace invariant by $A$ for any divisor $h'$ of $f$; 
\item all the $f$-GC codes are given by $\ker h''(A)$, where $h''$ is a divisor of $f$; 
\item If $\mathcal{C}$ is an $f$-GC code such that $\pi_f(\mathcal{C})=(g)$, then $n-\dim\mathcal{C}$ 
linearly independent columns of $h(A)$ form a basis of $\mathcal{C}^{\perp}$, where $f=h\cdot g$; moreover, if $b\in\mathbb{F}_q[X]$
is the smallest degree polynomial such that $\mathcal{C}\subseteq \ker b(A)$, then $b=ah$ for some $a\in\mathbb{F}_q\backslash \{0\}=:\mathbb{F}_q^*$.
\end{enumerate}
\end{coro}

\begin{proof}
Note that by Theorem \ref{ker skew}(1) and Corollary \ref{cor} we get cases (1),(3) and (2) respectively.
Let $\mathcal{C}$ be an $f$-GC code such that $\pi_f(\mathcal{C})=(g)$. Then from (1) it follows that $\mathcal{C}=\ker h(A)$, where
$f=h\cdot g$. Since $\mathrm{rk} h(A)=n-\dim\ker h(A)=n-\dim\mathcal{C}$, we deduce that $n-\dim\mathcal{C}$ 
linearly independent columns of $h(A)$ form a basis of $\mathcal{C}^{\perp}$. 
Finally, assume that $\mathcal{C}\subseteq \ker b(A)$ for some $b\in\mathbb{F}_q[X]$ with smallest degree. Then $h=bq+r$ for some $q,r\in\mathbb{F}_q[X]$ such that
$\deg r<\deg b$. Since $h(A)=b(A)q(A)+r(A)$, we see that $\mathcal{C}\subseteq\ker r(A)$. Thus $r=0$ and $h=bq$. 
Hence $f=b\cdot q\cdot g$ and this gives 
$(g)=\pi_f(\mathcal{C})\subseteq\pi_f(\ker b(A))=(qg).$ Therefore, we have $g=\alpha qg+\beta f$ for some $\alpha,\beta\in\mathbb{F}_q[X]$, i.e.
$1=(\alpha q+\beta b)q$. This shows that $q\in\mathbb{F}_q^*$ and that $b=ah$ with $a=q^{-1}$.
\end{proof}

The next useful result is related to polynomials in $N(R)$.

\begin{lem}\label{lemma1}
Let $a,b\in N(R)$ such that $a=bc=c'b$ for some $c,c'\in R$. Then $c,c'\in N(R)$.
\end{lem}

\begin{proof}
Assume that $a=bc$ for some $c\in R$. Then for any $d\in R$ we have 
$$b(cd)=(bc)d=ad=d'a=d'(bc)=(d'b)c=(bd'')c=b(d''c),$$ for some $d',d''\in R$, i.e. $cd=d''c$.
Similarly, one can prove that for any $e\in R$ there exists $e'\in R$ such that $c'e=e'c'$.
This shows that $c,c'\in N(R)$.
\end{proof}

From now on, assume that

\medskip

\noindent $(\#)$\quad $f=f _{1}^{\alpha _1}\cdot \ldots \cdot f _{t}^{\alpha
_t}$ is a factorization of $f$ as in $(*)$ into monic polynomials
$f_k\in N(R)$ which are irreducible in $ N(R)$ and such that $Rf_i\neq Rf_j$ for any $i\neq j$.

\medskip

\begin{obs}
If $\theta=id$, then we have $\delta_{\beta}^{\theta}=0$ and $(\#)$ always holds.
\end{obs}

On the other hand, when either $\theta\neq id$ or $\delta_{\beta}^{\theta}\neq 0$, we have the following

\begin{prop}\label{proposi}
$(\#)$ holds\ $\iff$ $f\in N(R)$. 
\end{prop}

\begin{proof}
If $(\#)$ holds, then $f\in N(R)$ because the product of elements in $N(R)$ belongs in $N(R)$. Assume now that $f\in N(R)$. By \cite[Theorem 9, p.38]{J} we know that $Rf=Rf_1^{\beta_1}\cdot ... \cdot Rf_s^{\beta_s}$ with $\beta_j\in\mathbb{Z}_{\geq 0}$, $f_j$ a monic polynomial in $N(R)$ and $Rf_j$ a maximal two sided ideal for any $j=1,...,s$. This shows that $f_i\neq f_j$ and that all the $f_i$'s are irreducible polynomials in $N(R)$. Thus $f=hf_1^{\beta_1}\cdot ... \cdot f_s^{\beta_s}$ for some $h\in R$. Since $f$ and all the $f_j$'s are polynomials in $N(R)$, from Lemma \ref{lemma1} it follows that $h\in N(R)$.
Then $f=f'h$ for some $f'\in R$ and this gives that $Rf=R(f'h)\subseteq Rh$. Hence $Rh$ is a two sided ideal containing $Rf$ and from \cite[p.38]{J} we deduce that $Rh=Rf_1^{\gamma_1}\cdot ... \cdot Rf_s^{\gamma_s}$, where $\gamma_i\in\mathbb{Z}_{\geq 0}$ and $\gamma_i\leq\beta_i$ for every $i=1,...,s$. Therefore we get that 
$$Rf_1^{\beta_1}\cdot ... \cdot Rf_s^{\beta_s}=Rf=R(h)\cdot R(f_1^{\beta_1})\cdot ... \cdot R(f_s^{\beta_s})=R(f_1)^{\beta_1+\gamma_1}\cdot ... \cdot R(f_s)^{\beta_s+\gamma_s},$$ and since the factorization of $Rf$ is unique, we conclude that $\gamma_i=0$ for any $i=1,...,s$, i.e. $Rh=R$. Hence $h\in\mathbb{F}_q^{\theta}$ and since $f$ is a monic polynomial, we obtain that $h=1$.
\end{proof}

\begin{lem}\label{lemma2}
Assume that $(\#)$ holds. Then for any $a_m,b_m\in\mathbb{Z}_{> 0}$ we have the following two properties:
\begin{enumerate}
\item[$(a)$] $\mathrm{lgcd}(\Pi_{k=1}^{r}f_{i_k}^{a_k}, \Pi_{h=1}^{s}f_{j_h}^{b_h})=1$, for any $\{j_1,...,j_s\}\subseteq \{1,...,t\}\backslash \{i_1,...,i_r\}$;
\item[$(b)$] $f_i^{a_i}f_j^{b_j}=f_j^{b_j}f_i^{a_i}$. 
\end{enumerate}
\end{lem}

\begin{proof}
$(a)$ For simplicity of notation write 
$g_i:=\Pi_{k=1}^{r}f_{i_k}^{a_k}$, $g_j:=\Pi_{h=1}^{s}f_{j_h}^{b_h}$ and $lgcd(g_i,g_j)=d_{ij}$, where $d_{ij}$ is a monic polynomial in $R$. 
Note that in fact $d_{ij}\in N(R)$. Since $Rg_i+Rg_j=Rd_{ij}$, we deduce that
there exist $a,b\in N(R)$ such that $g_i=ad_{ij}$ and $g_j=bd_{ij}$. Thus $Rg_i\subseteq Rd_{ij}$ and 
$Rg_j\subseteq Rd_{ij}$. By \cite[p.38]{J} we deduce that
$$Rd_{ij}=\Pi_{k=1}^{r}(Rf_{i_k})^{c_k}=\Pi_{h=1}^{s}(Rf_{j_h})^{d_h},$$ where $0\leq c_k\leq a_k$ and $0\leq d_h\leq b_h$.
From the uniqueness of the decomposition, we deduce that $c_k=0$ and $d_h=0$ for $k=1,...,r$ and $h=1,...,s$. Hence $Rd_{ij}=R$, that is,  
$d_{ij}=1$.

\medskip

\noindent $(b)$ Observe that the statement follows from $f_if_j=f_jf_i$. So, assume that $a_i=b_j=1$. From \cite[Lemma 4, p.38]{J} we deduce that $Rf_i\cdot Rf_j=Rf_j\cdot Rf_i$. This shows that
there exist $u,v\in R$ such that $f_if_j=uf_jf_i$ and $vf_if_j=f_jf_i$. Therefore we get $f_if_j=uf_jf_i=uvf_if_j$, i.e. $uv=1$.
So $u\in\mathbb{F}_q^*$ and $f_if_j=uf_jf_i$. Since all the $f_k$'s are monic polynomials, we see that $u=1$.
\end{proof}

\medskip

In line with \cite{RZ}, consider now the following subsets of $\mathbb{F}_q ^n$
\begin{equation}\label{(2)}
U_{i}:=\textrm{Ker} f_{i}^{\alpha _i}(T_f) 
\end{equation}
for $i=1,\ldots ,t$, where $T_f$ is as in $(**)$. Then we have the following properties for the $U_i$'s.

\begin{prop}\label{inv sub}
Under assumption $(\#)$, for the linear subspaces $U_{i}$ of $\mathbb{F}_q ^n$ we get the following properties:
\begin{enumerate}
\item[$(a)$] $U_i$ is a $(f,\Theta,\delta_{\beta}^{\theta})$-skew GC code; 
\item[$(b)$] $\dim U_i=\alpha_i\deg (f_i)$ for $i=1,...,t$;
\item[$(c)$] $\mathbb{F}_q ^n =U_1\oplus \cdots \oplus U_t$;
\item[$(d)$] if $\mathcal{C}$ is a $(f,\Theta,\delta_{\beta}^{\theta})$-skew GC code and $\mathcal{C}_i :=\mathcal{C}\cap U_i$ for $i=1,\ldots ,t$, 
then $\mathcal{C}_i$ is a $(f,\Theta,\delta_{\beta}^{\theta})$-skew GC code and $\mathcal{C}=\mathcal{C}_1 \oplus \cdots \oplus \mathcal{C}_t$;
\item[$(e)$] If $\alpha_i=1$ and $f_i$ is irreducible in $R$ for some $i\in \{1,...,t\}$, then $U_i$ is minimal with respect to the inclusion.
\end{enumerate}
\end{prop}

\begin{proof}
$(a)$ Since $f_{i}^{\alpha _i}\in N(R)$, from Corollary \ref{cor} we know that $\textrm{Ker} f_{i}^{\alpha _i}(T_f)$ is a 
$(f,\Theta,\delta_{\beta}^{\theta})$-skew GC code.

\smallskip

$(b)$ By Theorem \ref{ker skew}(1) and Lemma \ref{lemma2} we have $R\left(\frac{f}{f_i^{\alpha_i}} \right)=\pi_f(\mathrm{Ker}\ f_i^{\alpha_i}(T_f))$. 
Thus we deduce that
$$\dim U_i=\dim \mathrm{Ker}\ f_i^{\alpha_i}(T_f)=\deg f - \deg \left(\frac{f}{f_i^{\alpha_i}}\right) = \deg f_i^{\alpha_i}=\alpha_i\deg f_i.$$

\smallskip

$(c)$ Since $R$ is a principal ideal domain and the sum of two sided ideal is a two sided ideal, write $Rm=R\widehat{f }_1+ \dots + R\widehat{f }_t$ for some $m\in N(R)$, where $\widehat{f }_i:=\frac{f }{f _{i}^{\alpha _i}}$ for every $i=1,...,t$. Moreover, by Lemma \ref{lemma1} note that $\widehat{f }_i=a_im$ for some $a_i\in N(R)$ and every $i=1,...,t$.
Since
$$(Rf_1)^{\alpha_1}\cdot ... \cdot (Rf_{i-1})^{\alpha_{i-1}}\cdot (Rf_{i+1})^{\alpha_{i+1}}\cdot ... \cdot (Rf_t)^{\alpha_t}=R\widehat{f }_i\subset Rm,$$
from \cite[p. 38]{J} it follows that 
$$Rm=(Rf_1)^{\beta_1}\cdot ... \cdot (Rf_{i-1})^{\beta_{i-1}}\cdot (Rf_{i+1})^{\beta_{i+1}}\cdot ... \cdot (Rf_t)^{\beta_t},$$ where $0\leq\beta_k\leq\alpha_k$ for $k=1,...,i-1,i+1,...,t$. So we have
$$Rm =(Rf_2)^{\beta_2} ... (Rf_t)^{\beta_t}=(Rf_1)^{\beta_1}(Rf_3)^{\beta_3} ... (Rf_t)^{\beta_t}=...= (Rf_1)^{\beta_1} ... (Rf_{t-1})^{\beta_{t-1}},$$
and by Theorem 9 of \cite[p. 38]{J} we deduce that $Rm=R$, i.e. $R=R\widehat{f }_1+ \dots + R\widehat{f }_t$.
Thus $1=a_1\widehat{f }_1+ \dots + a_t\widehat{f }_t$ for some $a_1,...,a_t\in R$.
This gives $X^h=(X^ha_1)\widehat{f }_1+ \dots + (X^ha_t)\widehat{f }_t$ for every $h=0,...,n-1$. Therefore, for any $\vec{v}=(v_0,...,v_{n-1})\in\mathbb{F}^n_q$, we get
$$\vec{v}=\sum_{i=0}^{n-1}v_i\pi_f^{-1}(X^i)=\sum_{i=0}^{n-1}v_i\pi_f^{-1}\left( \sum_{j=1}^{t} (X^ia_j)\widehat{f }_j \right)=$$
$$=\sum_{j=1}^{t}\pi_f^{-1}\left[\left(\sum_{i=0}^{n-1}v_i X^i\right)a_j\widehat{f }_j \right]\in \pi_f^{-1}\left(R\widehat{f }_1\right)+ \dots + \pi_f^{-1}\left(R\widehat{f }_t\right),$$ i.e. $\vec{v}\in U_1+ \dots + U_t$ by Theorem \ref{ker skew}(1). Hence $\mathbb{F}_q ^n =U_1 + \cdots + U_t$.

Finally, observe that $R\widehat{f }_i\cap R\widehat{f }_j=RM_{ij}$ with $M_{ij}\in N(R)$ for any $i\neq j$. Moreover, $M_{ij}=a_i\widehat{f }_i=a_j\widehat{f }_j$ for some $a_i,a_j\in N(R)$. This gives $a_if_j^{\alpha_j}=a_jf_i^{\alpha_i}$ and then $Ra_i\cdot (Rf_j)^{\alpha_j}=Ra_j\cdot (Rf_i)^{\alpha_i}$.
From \cite[Theorem 9, p.38]{J} it follows that $Ra_i=Rh\cdot (Rf_i)^{\alpha_i}$ for some $h\in R$, i.e. $a_i=af_i^{\alpha_i}$ for some
$a\in N(R)$ by Lemma \ref{lemma1}. Therefore, $M_{ij}=af$ and we conclude that 
$$U_i\cap U_j= \pi_f^{-1}(R\widehat{f }_i\cap R\widehat{f }_j)=\pi_f^{-1}(R(af))\subseteq \pi_f^{-1}(Rf)=\vec{0}$$ for any $i\neq j$, that is, $\mathbb{F}_q ^n =U_1\oplus \cdots \oplus U_t$.

\smallskip

$(d)$ From (a), it follows that each $\mathcal{C}_i :=\mathcal{C}\cap U_i$ is a $(f,\Theta,\delta_{\beta}^{\theta})$-skew GC code of $\mathbb{F}_q ^n$.
Furthermore, from (c) we can deduce that 
$$\mathcal{C}=\mathcal{C}\cap\mathbb{F}_q ^n=\mathcal{C}\cap\left(U_1 \oplus \cdots \oplus U_t\right)
=\left(\mathcal{C}\cap U_1\right) \oplus \cdots \oplus \left(\mathcal{C}\cap U_t\right)=\mathcal{C}_1 \oplus \cdots \oplus \mathcal{C}_t.$$

$(e)$ For simplicity, assume that $i=1$. Let $U$ be a $(f,\Theta,\delta_{\beta}^{\theta})$-skew GC code such that
$$\{\vec{0}\}\subseteq U\subseteq U_1:=\textrm{Ker} f_{1}(A).$$ Then by Proposition \ref{corr. ideal2} and Theorem \ref{ker skew}(1) we know that there exists a right divisor
$g'\in R$ of $f$ such that $f=h'g'$ and 
$Rg'=\pi_f(U)\subseteq \pi_f(U_1)=R\frac{f}{f_1}$. Hence there is a polynomial $q\in R$ such that $g'=q\cdot f _{2}^{\alpha _1}\cdot \ldots \cdot f _{t}^{\alpha
_t}$ and this gives
$$f _{1}\cdot f _{2}^{\alpha _2}\cdot \ldots \cdot f _{t}^{\alpha_t}=f=h'g'=(h' q)\cdot f _{2}^{\alpha _1}\cdot \ldots \cdot f _{t}^{\alpha
_t},$$ i.e. $f _{1}=h' q$. Since $f_1$ is irreducible in $R$, we conclude that either $\deg q = 0$ or $q=f_1$. Therefore, we have either $U=\pi_f^{-1}(Rg')=\pi_f^{-1}(R\frac{f}{f_1})=U_1$ or
$U=\pi_f^{-1}(Rg')=\pi_f^{-1}(Rf)=\{\vec{0}\}$.
\end{proof}

\begin{prop}\label{f}
Let $f\in R$. Then $f\in N(R) \iff f(T_f)=0$.
\end{prop}

\begin{proof}
Assume that $f\in N(R)$. Then for every $\vec{v}=\pi_f^{-1}(v)\in\mathbb{F}_q^n$, from Remark \ref{rem 1bis} we conclude that
$$\pi_f(f(T_f)\vec{v})=f\pi_f(\vec{v})=fv=v'f=0\in R/Rf,$$ for some $v'\in R$, i.e. $f(T_f)(\vec{v})=\vec{0}$. Finally, suppose that 
$f(T_f)=0$. Consider $g\in R$. If $\deg g < \deg f$, then 
$$f\cdot g=f\cdot \pi_f (\pi_f^{-1}(g))=\pi_f (f(T_f)\pi_f^{-1}(g))=0\in R/Rf,$$ i.e. there exists $k\in R$ such that $f\cdot g=k\cdot f$.
On the other hand, if $\deg f \leq \deg g$ then there are $q,r\in R$ such that $g=qf+r$ with $\deg r < \deg f$. Thus by the above argument,
we can conclude that $f\cdot g=f\cdot (qf+r)=fqf+fr=fqf+r'f=(fq+r')f$ for some $r'\in R$. This shows that $f\in N(R)$.
\end{proof}

\smallskip

\begin{obs}
When $\delta_{\beta}=0$, then $f(T_f)=0$ if and only if $f\in\mathbb{F}_q^{\theta}[X^k;\theta, 0]$ with $k$ the orden of $\theta$.
If $\delta_{\beta}\neq 0$, then {\em Proposition \ref{f}} gives a criterion to know when $f\in N(R)$.
\end{obs}

Having in mind Proposition \ref{proposi}, we have the following

\begin{coro}\label{ff}
$(\#)$ holds $\iff$ $f(T_f)=0$.
\end{coro}

Let $f$ be as in $(\#)$ and write $f=f_i^{\alpha_i}\cdot \widehat{f }_i$. By Lemma \ref{lemma2} we know that $\mathrm{lgcd}(f_i^{\alpha_i}, \widehat{f }_i)=1$. Thus there exist 
$a_i,b_i\in R$ such that $a_if_i^{\alpha_i}+b_i\widehat{f }_i=1$. So we get $(b_i\widehat{f }_i)(a_if_i^{\alpha_i})+(b_i\widehat{f }_i)^2=b_i\widehat{f }_i$, i.e.
$(b_i\widehat{f }_i)^2+c_if=b_i\widehat{f }_i,$ for some $c_i\in R$. Therefore by Lemma 2.7 of \cite{BL} and Corollary \ref{ff} we deduce that 
$$b_i(T_f)\widehat{f }_i(T_f)=(b_i(T_f)\widehat{f }_i(T_f))^2+c_i(T_f)f(T_f)=(b_i(T_f)\widehat{f }_i(T_f))^2.$$
As in \cite{RZ}, define $e_i(T_f):=b_i(T_f)\widehat{f }_i(T_f)$. In the same spirit of \cite{RZ} and for the convenience of the reader, we
give and prove the following result.

\begin{prop}\label{inv sub2}
Under assumption $(\#)$, we get the following properties:
\begin{enumerate}
\item[$(a)$] $e_i(T_f)^2=e_i(T_f)$;
\item[$(b)$] $e_i(T_f)e_j(T_f)=0$ for $i\neq j$;
\item[$(c)$] $e_i(T_f)(\vec{v}_j)=0$ for every $\vec{v}_j\in U_j$ with $j\neq i$;
\item[$(d)$] $\vec{v}\in U_i \iff e_i(T_f)(\vec{v})=\vec{v}$; moreover, if $T$ is an idempotent endomorphism of $\mathbb{F}_q^{n}$ such that 
$\vec{v}\in U_i \iff T(\vec{v})=\vec{v}$, then $T=e_i(T_f)$;
\item[$(e)$] $\sum_{i=1}^{t} e_i(T_f)=id$;
\item[$(f)$] If $\theta =id$, then $U_i=<e_i(A)_1,..., e_i(A)_n>$, where $e_i(A)_j$ is the $j^{th}$ row of $e_i(A)$. 
\end{enumerate}
\end{prop}

\begin{proof}
$(a)$ It follows from the definition of $e_i(T_f)$.

$(b)$ By \cite[Lemma 2.7]{BL} we deduce that
$$e_i(T_f)e_j(T_f)=(e_ie_j)(T_f)=(b_i\widehat{f }_ib_j\widehat{f }_j)(T_f)=(cf)(T_f)=c(T_f)f(T_f)=0,$$
for some $c\in R$.

$(c)$ Let $\vec{u_j}\in U_j$. Then there is a $t\in R$ such that 
$$e_i(T_f)(\vec{u_j})=(b_i\widehat{f }_i)(T_f)(\vec{u_j})=(tf_j^{\alpha_j})(T_f)(\vec{u_j})=
t(T_f)f_j^{\alpha_j}(T_f)(\vec{u_j})=0.$$

$(d)$ Let $\vec{v}\in U_i$. Note that 
$$a_i(T_f)f_i^{\alpha_i}(T_f)+e_i(T_f)=a_i(T_f)f_i^{\alpha_i}(T_f)+b_i(T_f)\widehat{f }_i(T_f)=id.$$
Hence $e_i(T_f)(\vec{v})=a_i(T_f)f_i^{\alpha_i}(T_f)(\vec{v})+e_i(T_f)(\vec{v})=\vec{v}$. On the other hand, 
if $e_i(T_f)(\vec{v})=\vec{v}$ then there exists a $c\in R$ such that
$$f_i^{\alpha_i}(T_f)(\vec{v})=(f_i^{\alpha_i}e_i)(T_f)(\vec{v})=(f_i^{\alpha_i}b_i\widehat{f }_i)(T_f)(\vec{v})
=(cf)(T_f)(\vec{v})=c(T_f)f(T_f)(\vec{v})=0,$$ i.e. $\vec{v}\in U_i$. 

Let $T$ be an idempotent endomorphism, that is $T^2=T$, such that $\vec{v}\in U_i \iff T(\vec{v})=\vec{v}$.
Then $Im(T)=U_i$ and for every $\vec{v}\in\mathbb{F}_q^n$ we can write $\vec{v}=[\vec{v}-T(\vec{v})]+T(\vec{v})$, where
$\vec{v}-T(\vec{v})\in\ker T$ and $T(\vec{v})\in Im(T)$. Note that $Im(T)\cap\ker(T)=\vec{0}$ since $T^2=T$. Thus by Proposition 
\ref{inv sub} we see that
$$\mathbb{F}_q^n=U_1\oplus ... \oplus U_t=Im(T)\oplus\ker (T)=U_i\oplus\ker (T),$$ i.e. $\ker (T)=U_1\oplus ... \oplus U_{i-1}
\oplus U_{i+1}\oplus ... \oplus U_t$. Then for any $\vec{v}\in\mathbb{F}_q^n$ we have $\vec{v}=\vec{v}_1+...+\vec{v}_t$, with $\vec{v}_j\in U_j$, and
by $(c)$ we can conclude that
$$T(\vec{v})=T(\vec{v}_1+...+\vec{v}_t)=T(\vec{v}_1)+...+T(\vec{v}_t)=T(\vec{v}_i)=\vec{v}_i=e_i(T_f)\vec{v}_i=e_i(T_f)(\vec{v}),$$
i.e. $T=e_i(T_f)$.

$(e)$ For every $\vec{v}\in\mathbb{F}_q^n$, we have $\vec{v}=\vec{v}_1+...+\vec{v}_t$ with $\vec{v}_i\in U_i$ for any $i=1,...,t$. Thus by $(c)$ and $(d)$ we
obtain that
$$\left( \sum_{i=1}^{t} e_i(T_f)\right)(\vec{v})=\sum_{i=1}^{t}e_i(T_f)(\vec{v}_i)=\sum_{i=1}^{t} \vec{v}_i=id(\vec{v}).$$

$(f)$ If $\vec{u}_i\in U_i$, then from $(d)$ it follows that $\vec{u}_i=\vec{u}_ie_i(A)\in<e_i(A)_1,..., e_i(A)_n>$, i.e. $U_i\subseteq <e_i(A)_1,..., e_i(A)_n>$. 
Moreover, note that there exists a polynomial $s\in R$ such that $e_i(A)\cdot f_i^{\alpha_i}(A)=(e_if_i^{\alpha_i})(A)=(sf)(A)=s(A)f(A)=0$, that is, all the rows of $e_i(A)$ 
belong to $U_i$. Hence $<e_i(A)_1,..., e_i(A)_n>\subseteq U_i$, i.e. $U_i=<e_i(A)_1,..., e_i(A)_n>$.
\end{proof}

\medskip

\section{Further properties of skew GC codes with $\delta_{\beta}^{\theta}=0$}

In this last section, we give some further results when the derivation $\delta_{\beta}^{\theta}$ is zero, e.g., when either $\beta =0$, or $\theta=id$.

\subsection{The minimal polynomial of a semi-linear transformation}

From Proposition \ref{f} we know that $f\in N(R)$ if and only if $f(T_f)=0$. In this subsection we show how to construct the minimal polynomial of any semi-linear transformation $T:=\Theta\circ M$ defined over $\mathbb{F}_q^n$.

\begin{lem}\label{lemma}
Let $T:=\Theta\circ M$ be a $\theta$-semi-linear transformation on $\mathbb{F}^n_q$. Then
$$T\cdot\lambda = \theta(\lambda)\cdot T\qquad \forall \lambda\in \mathbb{F}_q .$$
\end{lem}

\begin{proof}
Take a vector $\vec{v}\in\mathbb{F}_q^n$ and an element $\lambda\in \mathbb{F}_q$. Then
$$(T\cdot\lambda)(\vec{v})=T(\lambda\vec{v})=(\lambda\vec{v})\Theta\circ M=((\vec{v})\Theta\circ M)\theta(\lambda)=\theta(\lambda)\left(T(\vec{v})\right)
=\left(\theta(\lambda)T\right)(\vec{v}),$$
that is, $T\cdot\lambda =\theta(\lambda)\cdot T$ for any $\lambda\in \mathbb{F}_q$.
\end{proof}

\noindent Note that by Lemma \ref{lemma} one can consider the surjective ring homomorphism
$$\sigma :\mathbb{F}_q[X;\theta]\rightarrow \mathbb{F}_q[T;\theta],$$ 
defined by $p\mapsto p(T)$, where $\mathbb{F}_q[Z;\theta]:=\mathbb{F}_q[Z;\theta,0]$.

\smallskip

First of all, let us show that for any semi-linear transformation $T=\Theta \circ M$, there exists always a unique monic minimal polynomial $m_T\in \mathbb{F}_q[X;\theta]$ 
such that $m_T(T)=O$, where $O:\mathbb{F}_q^n\to\mathbb{F}_q^n$ is the null-map. 

Let $k$ be the order of $\theta$, that is, $\theta^k=id$. Then by \cite[Lemma 4]{TT} we have
\begin{equation*}
\begin{split}
T^k&=(\Theta \circ M)^k\\
&=\Theta ^kM_{\theta ^{k-1}}M_{\theta ^{k-2}}\ldots M_{\theta}M\\
&=id \circ B=B,
\end{split}
\end{equation*}
where $B:=M_{\theta ^{k-1}}M_{\theta ^{k-2}}\ldots M_{\theta}M$ is a matrix with coefficients in $\mathbb{F}_q$. Therefore, there exists a (minimal) polynomial 
$m_B=X^h+b_{h-1}X^{h-1}+\ldots +b_1X+b_0 \in \mathbb{F}_q[X]$ such that
$$m_B(B)=O=(T^k)^h+b_{h-1}(T^k)^{h-1}+\ldots +b_1(T^k)+b_0(id).$$
This gives a polynomial $m:=m_B(X^k)\in\mathbb{F}_q[X;\theta]$ such that $m(T)=O$, showing the existence of a unique minimal monic polynomial $m_T\in\mathbb{F}_q[X;\theta]$ such that $m_T(T)=O$. In fact, we can prove the following 

\begin{prop}\label{propos}
Let $T:=\Theta\circ M$ be a $\theta$-semi-linear transformation on $\mathbb{F}^n_q$ such that $M$ is an invertible $n\times n$ matrix.
Then the unique minimal monic polynomial $m_T\in \mathbb{F}_q[X;\theta]$ 
such that $m_T(T)=O$ is given by $m_T=m_B(X^k)\in\mathbb{F}_q^{\theta}[X;\theta]$, where $k$ is the order of $\theta$, $m_B$ is the minimal monic polynomial of 
the $n\times n$ matrix $B:=M_{\theta ^{k-1}}M_{\theta ^{k-2}}\ldots M_{\theta}M$ and $\mathbb{F}_q^{\theta}$ is the field fixed by $\theta$.
\end{prop}

\begin{proof}
Let $m_T\in \mathbb{F}_q[X;\theta]$ be the unique minimal monic polynomial such that $m_T(T)=O$. Take any polynomial $a\in \mathbb{F}_q[X;\theta]$
and write $a\cdot m_T=m_T\cdot q+r$, where $r$ is the reminder of the left division of $a\cdot m_T$ by $m_T$. Since $\sigma$ is a ring homomorphism, we have
$$O=a(T)\cdot m_T(T)=\sigma (a)\cdot \sigma(m_T)=\sigma(a\cdot m_T)=\sigma (m_T\cdot q+r)=\sigma (m_T\cdot q) + \sigma (r) = $$
$$=\sigma (m_T)\cdot \sigma (q) + \sigma (r)=m_T(T)\cdot q(T) + r(T) = r(T).$$
By the minimality of $m_T$, we deduce that $r=0$. Then $a\cdot m_T=m_T\cdot q$ for some $q\in \mathbb{F}_q[X;\theta]$. This shows that $Rm_T=m_TR=(m_T)$, 
i.e. $m_T$ is an invariant polynomial in $\mathbb{F}_q[X;\theta]$. So $m_T=X^t\cdot b=b\cdot X^t$ for an integer $t\in\mathbb{Z}_{\geq 0}$ and some 
$b\in\mathbb{F}_q^{\theta}[X^k;\theta]$.
Since $M$ is an invertible $n\times n$ matrix, we deduce that $T$ is an invertible map. Hence $O=m_T(T)=T^t b(T)$ implies $b(T)=O$, i.e. $m_T=b$ with
$\deg m_T=sk$ for some $s\in\mathbb{Z}_{\geq 1}$ by the minimality of $m_T$.
Since $T^k=B$, we conclude that $O=m_T(T)=p(B)$ for a polynomial $p\in\mathbb{F}_q^{\theta}[X;\theta]$ with $s=\deg p\geq\deg m_B$. Thus 
we get $m_T(T)=m_B(X^k)\in\mathbb{F}_q^{\theta}[X;\theta]$.
\end{proof}

This allows us to obtain the following

\begin{teo}\label{matrixpol3}
Let $T:=\Theta\circ M$ be a $\theta$-semi-linear transformation on $\mathbb{F}_q^n$ such that $M$ is an invertible $n\times n$ matrix.
Then there exists a ring isomorphism
$$\mathbb{F}_q[X;\theta]/(m_B(X^k))\cong \mathbb{F}_q[T;\theta]$$
defined by $[p]\mapsto p(T)$, where $[p]$ is the class of a polynomial $p\in\mathbb{F}_q[X;\theta]$,
$k$ is the order of $\theta$ and $m_B\in\mathbb{F}_q^{\theta}[X;\theta]$ is the monic minimal polynomial of 
the $n\times n$ matrix $B:=M_{\theta ^{k-1}}M_{\theta ^{k-2}}\ldots M_{\theta}M$.
\end{teo}

\begin{proof}
Consider the ring homomorphism $\sigma :\mathbb{F}_q[X;\theta]\rightarrow \mathbb{F}_q[T;\theta]$, defined by $p\mapsto p(T)$. By construction, $\sigma$ is surjective. 
Moreover, by Proposition \ref{propos} note that $\ker (\sigma )=(m_T)$. Then there exists an isomorphism $\overline{\sigma }$ between $\mathbb{F}_q[X;\theta]/(m_T)$ and $\mathbb{F}_q[T;\theta]$, where $\overline{\sigma }$ is defined by $[p]\mapsto p(T)$ and $[p]$ is the class of a polynomial $p$.
\end{proof}

\begin{obs}\label{rem1}
If $p=q$ in $\mathbb{F}_q[X;\theta]$, then $[p]=[q]$ in $\mathbb{F}_q[X;\theta]/(m_B(X^k))$. Thus $p(T)=q(T)$ via $\overline{\sigma }$.
Moreover, when $\theta = id$, from {\em Theorem \ref{matrixpol3}} we deduce that there exists a ring isomorphism
$$\mathbb{F}_q[X]/(m_M)\cong \mathbb{F}_q[M]$$
defined by $[p]\mapsto p(M)$, where $[p]$ is the class of a polynomial $p\in\mathbb{F}_q[X]$ 
and $m_M\in\mathbb{F}_q[X]$ is the monic minimal polynomial of the $n\times n$ matrix $M$. 
\end{obs}

\subsection{BCH lower bounds for the minimum distance}

In this subsection we show some results which give lower bounds for the distance of a skew GC-code.

Assume that 
$$f=X^{n}-f_{n-1}X^{n-1}-\ldots -f_1X-f_0,$$ where $f_{n-1},\ldots ,f_1,f_0 \in \mathbb{F}_q$ and $f_0\neq 0$.

\begin{lem}\label{inverse}
In $R/Rf$ we have
$$\alpha\cdot X=1\qquad\mathrm{and}\qquad X\cdot\beta=1,$$
where $$\alpha:=f_0^{-1}X^{n-1}-f_0^{-1}f_{n-1}X^{n-2}-\ldots -f_0^{-1}f_{2}X-f_0^{-1}f_1$$ and 
$$\beta:=\theta^{-1}(f_0^{-1})X^{n-1}-\theta^{-1}(f_0^{-1}f_{n-1})X^{n-2}-\ldots -\theta^{-1}(f_0^{-1}f_1).$$
In particular, when $\theta=id$, we get $X\cdot \alpha=\alpha\cdot X=1$.
\end{lem}

\begin{proof} It is sufficient to note that in $R/Rf$ we have the following equivalences:
\begin{equation*}
\begin{split}
&X^{n}-f_{n-1}X^{n-1}-\ldots -f_1X-f_0=0 \iff \\
& \iff (X^{n-1}-f_{n-1}X^{n-2}-\ldots -f_1)\cdot X=f_0\\
& \iff (f_0^{-1}X^{n-1}-f_0^{-1}f_{n-1}X^{n-2}-\ldots -f_0^{-1}f_{2}X-f_0^{-1}f_1)\cdot X=1
\end{split}
\end{equation*}
and
\smallskip
\begin{equation*}
\begin{split}
&X^{n}-f_{n-1}X^{n-1}-\ldots -f_1X-f_0=0 \\
& \iff f_0^{-1}X^{n}-f_0^{-1}f_{n-1}X^{n-1}-\ldots -f_0^{-1}f_1X=1 \\
& \iff X\cdot (\theta^{-1}(f_0^{-1})X^{n-1}-\theta^{-1}(f_0^{-1}f_{n-1})X^{n-2}-\ldots -\theta^{-1}(f_0^{-1}f_1))=1. \\
\end{split}
\end{equation*}
\end{proof}

Let $\vec{v}\in \mathbb{F}_q^n$. We will denote by $\mathrm{wt}(v)$ the Hamming weigth of $\vec{v}$, where $\pi _f(\vec{v})=v \in R/Rf$.

\begin{lem}\label{vtv}
Let $\mathcal{C}\subset \mathbb{F}_q^n$ be a skew GC code, and consider a polynomial $c\in \pi _{f}(\mathcal{C} )$ with weight $wt(c)=w$. 
Then, there exists $h\in R$ such that
$$h\cdot c=1+\displaystyle \sum _{i=1}^{w-1}c_iX^{a_i} \in \pi _f(\mathcal{C}),$$
where $c_i \in \mathbb{F}_q^*$ and $a_i \in \mathbb{N}$ with $a_i\leq n-1$ for $i=1,\ldots ,w-1$. Furthermore,
$$\mathrm{wt}(h\cdot c)=\mathrm{wt}\left( 1+\displaystyle \sum _{i=1}^{w-1}c_iX^{a_i}\right)=w.$$
\end{lem}

\begin{proof}
Since $\mathrm{wt}(c)=w$, we can write $c=b_{i_0}X^{i_0}+b_{i_1}X^{i_1}+\ldots +b_{i_{w-1}}X^{i_{w-1}}$, where $i_0<\ldots <i_{w-1}$ and $b_{i_j}\neq 0$ for $j=0,\ldots ,w-1$. Hence
$$c=b_{i_0}X^{i_0}(1+\theta^{-i_0}(b_{i_0}^{-1}b_{i_1})X^{i_1-i_0}+\ldots +\theta^{-i_0}(b_{i_0}^{-1}b_{i_{w-1}})X^{i_{w-1}-i_0}).$$
Since $\mathcal{C}$ is a skew GC code, we can conclude that
$$\alpha^{i_0}b_{i_0}^{-1}c=1+\theta^{-i_0}(b_{i_0}^{-1}b_{i_1})X^{i_1-i_0}+\ldots +\theta^{-i_0}(b_{i_0}^{-1}b_{i_{w-1}})X^{i_{w-1}-i_0} \in \pi _f(\mathcal{C} ).$$
The statement follows by putting $h=\alpha^{i_0}b_{i_0}^{-1}$, $c_j=\theta^{-i_0}(b_{i_0}^{-1}b_{i_j})$ and $a_j=i_j-i_0$ for $j=0,\ldots ,w-1$.
\end{proof}

\begin{defi}
Let $\mathcal{C}\subset \mathbb{F}_q$ be a skew GC code. The distance $d_{\mathcal{C}}$ of $\mathcal{C}$ is defined as
$$d_{\mathcal{C}}:=\min \{d(\vec{x},\vec{y}):\ \vec{x},\vec{y} \in \mathcal{C}, \ \vec{x}\neq \vec{y} \},$$
where $d$ is the Hamming distance.
\end{defi}

\begin{obs}\label{ordfinite}
Consider an element $\beta$ in $\overline{\mathbb{F}_q}$ such that $p(\beta ^k)=0$, for some $p\in R$ and $k\in \mathbb{Z}_{>0}$. Then, from \cite{TT} it follows that there exists $q\in R$ such that
$$\beta ^m-1=q(\beta )p(\beta ),$$
i.e. $\beta ^m=1$ for some $m\in \mathbb{N}^*$. Hence, $\textrm{ord}( \beta )<+\infty$, where $\textrm{ord}(\beta )$ denotes the order of $\beta$.
\end{obs}

Inspired by \cite{HT}, the following results provide lower bounds on the minimum Hamming distance for a skew GC code.

\begin{teo}\label{bound1}
Let $\mathcal{C}=(g)_{n,q}^{\theta,0}$ be a skew GC code. Suppose there exists $\beta \in \overline{\mathbb{F}_q}$, $l\in \mathbb{Z}_{\geq 0}$, $c\in \mathbb{Z}_{>0}$ such that $g(\beta ^{l+ci})=0$ for $i=0,\ldots ,\delta-2$. If $N_i(\beta ^{c})\neq 1$ for every $i=1,...,n-1$, then $d_{\mathcal{C}} \geq \delta$.
\end{teo}

\begin{proof}
Suppose there exists a polynomial $c\in \pi _f(\mathcal{C})$ with $wt(c)=w<\delta$. By Lemma \ref{vtv}, we can assume that $c$ is the following type
$$c=1+\displaystyle \sum _{i=1}^{w-1}c_iX^{a_i}$$
where $c_i\in \mathbb{F}_q^*$ and $a_i\in \mathbb{Z}_{>0}$ with $a_i<n$ for $i=1,\ldots ,w-1$.
Define
$$Y_i:=N_{a_i}(\beta)\qquad \mbox{,}\qquad S_j:=\displaystyle \sum _{i=1}^{w-1}c_iY_i^j=c(\beta ^j)-1.$$
and
\begin{equation*}
\begin{split}
p&:=\displaystyle \prod _{i=1}^{w-1}(X-Y_i^c)\\
&=X^{w-1}+p_1X^{w-2}+\ldots +p_{w-2}X+p_{w-1} \in \mathbb{F}_q[X].
\end{split}
\end{equation*}
Since by hypothesis $Y_i^c=N_{a_i}(\beta^c)\neq 1$, we deduce that $p(1)\neq 0$.\\

By arguing as in \cite{HT}, we have the following equalities:
\begin{equation*}
\begin{split}
0&=\displaystyle \sum _{i=1}^{w-1}c_iY_i^lp(Y_i^c)\\
&=\sum _{i=1}^{w-1}c_iY_i^l\left( Y_i^{c(w-1)}+p_1Y_i^{c(w-2)}+\ldots +p_{w-2}Y_i^c+p_{w-1}\right) \\
&=\sum _{i=1}^{w-1}c_i\left( Y_i^{l+c(w-1)}+p_1Y_i^{l+c(w-2)}+\ldots +p_{w-2}Y_i^{l+c}+p_{w-1}Y_i^l\right) \\
&=\sum _{i=1}^{w-1}c_iY_i^{l+c(w-1)}+p_1\left( \sum _{i=1}^{w-1}c_iY_i^{l+c(w-2)}\right) +\ldots +p_{w-2}\left( \sum _{i=1}^{w-1}c_iY_i^{l+c}\right) + \\
&+p_{w-1}\left( \sum _{i=1}^{w-1}c_iY_i^l\right)=S_{l+c(w-1)}+p_1S_{l+c(w-2)}+\ldots +p_{w-2}S_{l+c}+p_{w-1}S_l
\end{split}
\end{equation*}
i.e. \begin{center}
$S_{l+c(w-1)}+p_1S_{l+c(w-2)}+\ldots +p_{w-2}S_{l+c}+p_{w-1}S_l=0.$\hfill{($\ast$)}
\end{center}

Since $c\in \pi _f(\mathcal{C})$, we have $c(\beta ^{l+ci})=0$ for $i=0,\ldots ,\delta -2$. Then
$$S_{l+ci}=c(\beta ^{l+ci})-1=-1$$
for  every $i=0,\ldots ,\delta -2$. So by $(\ast )$ we conclude
$$0=S_{l+c(w-1)}+p_1S_{l+c(w-2)}+\ldots +p_{w-2}S_{l+c}+p_{w-1}S_l=-p(1),$$
i.e. $p(1)=0$, a contradiction. Hence $d_{\mathcal{C}}\geq \delta$.
\end{proof}

\begin{obs}
The condition $\mathrm{rk}\ V_n(N_0(\beta^c),N_1(\beta^c),...,N_{n-1}(\beta^c))=n$ implies the hypothesis of the above result. Thus \emph{Theorem \ref{bound1}}
generalizes \emph{Theorem 4} of \cite{BU2} when the derivation is trivial.
\end{obs}

\begin{ejem}
Assume that $\theta=id$. Then the assumptions $\mathrm{ord}(\beta)\geq n$ and $\gcd (\mathrm{ord}(\beta ),c)=1$ imply the hypothesis $\beta^{a_ic}\neq 1$ of \emph{Theorem \ref{bound1}} 
for every $i=1,...,n-1$. Moreover, if only the hypothesis $\mathrm{ord}(\beta)\geq n$ holds, then the distance of the code may be less than the expected lower bound. 
For instance, consider the vector space $\mathbb{F}_7^6$ with $\beta =5$, $c=4$ and $l=1$. Then the GC code $\mathcal{C}=((X-5)(X-3))^{id,0}_{6,7}$ has distance $2<3$.
\end{ejem}

The following result is a generalization of the above theorem and in some circumstances gives a better lower bound than that of Theorem \ref{bound1}.

\begin{teo}\label{bound2}
Let $\mathcal{C}=(g)_{n,q}^{\theta,0}$ be a skew GC code. Suppose there exist $ \beta \in \overline{\mathbb{F}_q}$, $l,c_1,c_2 \in \mathbb{Z}_{\geq 0}$ such that $(c_1,c_2)\neq (0,0)$, $g(\beta ^{l+c_1i_1+c_2i_2})=0$ for $i_1=0,\ldots ,\delta -2$ and $i_2=0,\ldots ,s$. 
If $N_i(\beta ^{c_j})\neq 1$ for every $i=1,...,n-1$ and $j=1,2$, then we have $d_{\mathcal{C}} \geq \delta +s$.
\end{teo}

\begin{proof}
By Theorem \ref{bound1} it follows that, $d_{\mathcal{C}}\geq \delta$. Suppose there exist an element $c\in \pi _f(\mathcal{C})$ with $\mathrm{wt}(c)=w$ such that $\delta \leq w<\delta +s$. By Lemma \ref{vtv}, as in the Proof of Theorem \ref{bound1} write
$$c=1+\displaystyle \sum _{i=1}^{w-1}c_iX^{a_i}$$
where $c_i\in \mathbb{F}_q^*$ and $a_i\in \mathbb{Z}_{>0}$ with $a_i<n$ for $i=1,\ldots ,w-1$.

Similary to Theorem \ref{bound1}, define again
$$Y_i:=N_{a_i}(\beta)\qquad \mbox{,}\qquad S_j:=\displaystyle \sum _{i=1}^{w-1}c_iY_i^j=c(\beta ^j)-1$$
and
\begin{equation*}
\begin{split}
p&:=\displaystyle \prod _{i_1=1}^{\delta -2}(X-Y_{i_1}^{c_1})\\
&=X^{\delta -2}+p_1X^{\delta -3}+\ldots +p_{\delta -3}X+p_{\delta -2}\in \mathbb{F}_q[X],\\
\medskip
q&:=\displaystyle \prod _{i_2=\delta -1}^{w-1}(X-Y_{i_2}^{c_2})\\
&=X^{w-\delta +1}+q_1X^{w-\delta}+\ldots +q_{w-\delta}X+q_{w-\delta +1}\in \mathbb{F}_q[X],\\
\medskip
r&:=pq
\end{split}
\end{equation*}
Since $Y_{i_j}^{c_j}=N_{a_{i_j}}(\beta^{c_j})\neq 1$ for $j=1,2$, we see that $r(1)\neq 0$.\\

On the other hand, we have the following equalities:

\begin{equation*}
\begin{split}
\small 0&=\displaystyle \sum _{i=1}^{w-1}c_iY_i^lp(Y_i^{c_1})q(Y_i^{c_2})\\
&=\displaystyle \sum _{i=1}^{w-1}c_iY_i^l(Y_i^{c_1(\delta -2)}+p_1Y_i^{c_1(\delta -3)}+\ldots +p_{\delta -2})(Y_i^{c_2(w-\delta +1)}+q_1Y_i^{c_2(w-\delta )}+\\
&+\ldots +q_{w-\delta +1})\\
&=\displaystyle \sum _{i=1}^{w-1}c_i[(Y_i^{l+c_1(\delta -2)}+p_1Y_i^{l+c_1(\delta -3)}+\ldots +p_{\delta -2}Y_i^l)(Y_i^{c_2(w-\delta +1)}+q_1Y_i^{c_2(w-\delta )}+\\
&+\ldots +q_{w-\delta +1})]\\
&=\displaystyle \sum _{i=1}^{w-1} c_i [(Y_i^{l+c_1(\delta -2)+c_2(w-\delta +1)}+p_1Y_i^{l+c_1(\delta -3)+c_2(w-\delta +1)}+\ldots +p_{\delta -2}Y_i^{l+c_2(w-\delta +1)})+
\end{split}
\end{equation*}
\begin{equation*}
\begin{split}
&+q_1(Y_i^{l+c_1(\delta -2)+c_2(w-\delta )}+p_1Y_i^{l+c_1(\delta -3)+c_2(w-\delta )}+\ldots +p_{\delta -2}Y_i^{l+c_2(w-\delta )})+\ldots +\\
&+q_{w-\delta +1}(Y_i^{l+c_1(\delta -2)}+p_1Y_i^{l+c_1(\delta -3)}+\ldots +p_{\delta -2}Y_i^l)]\\
&=(S_{l+c_1(\delta -2)+c_2(w-\delta +1)}+p_1S_{l+c_1(\delta -3)+c_2(w-\delta +1)}+\ldots +p_{\delta -2}S_{l+c_2(w-\delta +1)} )+\ldots +\\
&+q_1(S_{l+c_1(\delta -2)+c_2(w-\delta )}+p_1S_{l+c_1(\delta -3)+c_2(w-\delta )}+\ldots +p_{\delta -2}S_{l+c_2(w-\delta )})+\ldots +\\
&+q_{w-\delta +1}(S_{l+c_1(\delta -2)}+p_1S_{l+c_1(\delta -3)}+\ldots +p_{\delta -2}S_l) .
\end{split}
\end{equation*}

Since $c\in \pi _f(\mathcal{C} )$ we know that $c(\beta ^{l+c_1i_1+c_2i_2})=0$. Hence
$$S_{l+c_1i_1+c_2i_2}=c(\beta ^{l+c_1i_1+c_2i_2})-1=-1,$$
for $i_1=0,\ldots ,\delta -2$, $i_2=0,\ldots ,s$ and $\delta \leq w<\delta +s$. Therefore from the above equation and the inequalities, we can conclude
$$0=(1+q_1+\ldots +q_{w-\delta}+q_{w-\delta +1})(-1-p_1-\ldots -p_{\delta -3}-p_{\delta -2})=-r(1),$$
i.e. $r(1)=0$, but this give a contradiction. Hence $d_{\mathcal{C}}\geq \delta +s$.
\end{proof}

\begin{obs}
The condition $\mathrm{rk}\ V_n(N_0(\beta^{c_j}),N_1(\beta^{c_j}),...,N_{n-1}(\beta^{c_j}))=n$ for $j=1,2$ implies the hypothesis of \emph{Theorem \ref{bound2}}.
\end{obs}

\begin{obs}
Assume that $\theta=id$. Then the assumptions $\mathrm{ord}(\beta)\geq n$ and $\gcd (\mathrm{ord}(\beta ),c_j)=1$ for $j=1,2$ imply the hypothesis $\beta^{a_i{c_j}}\neq 1$ of \emph{Theorem \ref{bound2}} 
for every $i=1,...,n-1$ and $j=1,2$. Furthermore, we get the following properties:
\begin{enumerate}
\item[$(a)$] Let $X^m-1=qf$ be as in \cite[Lemma 26]{TT}. If $\gcd (m,\mathrm{Char}(\mathbb{F}_q))=1$, then there exist a primitive root $\beta \in \overline{\mathbb{F}_q}$ of $X^m-1$ such that $\mathrm{ord}(\beta )=m\geq n$. Moreover, if $g(\alpha )=0$ then $\alpha =\beta ^h$ for some $h\in \mathbb{Z}{\geq 0}$.
\medskip
\item[$(b)$] Let $f=X^n-1$. Suppose there exists a primitive root $\beta$ of $f$. Then all roots $\alpha$ of $g$ are of type $\alpha =\beta ^h$ for some $h\in \mathbb{Z}_{\geq 0}$. 
Moreover, the hypothesis of {\em Theorem \ref{bound1}} reduce to that of {\em Theorem 1} in \cite{HT} when $\theta =id$.
\end{enumerate}
\end{obs}

 Let us give now a natural extension of Theorem \ref{bound2}, which can be proved by an inductive argument.

\begin{teo}\label{bound3}
Let $\mathcal{C}=(g)_{n,q}^{\theta,0}$ be a skew GC code. Suppose there exist $\beta \in \overline{\mathbb{F}_q}$ and $l,c_1,\ldots ,c_r\in \mathbb{Z}_{\geq 0}$ such that 
$(c_1,\ldots ,c_r)\neq (0,\ldots ,0)$ and $g(\beta ^{l+\sum _{k=1}^ri_kc_k})=0$ for $i_1=0,\ldots ,\delta -2$,  $i_k=0,\ldots ,s_k$ and $k=2,\ldots, r$. 
If $N_i(\beta ^{c_j})\neq 1$ for every $i=1,...,n-1$ and $j=1,...,r$, then $d_{\mathcal{C}} \geq \delta +\sum _{k=2}^rs_k$.
\end{teo}

\begin{obs}
If $\beta \in \mathbb{F}_q$, then $\mathrm{ord} (\beta )\leq q-1$ and in {\em Theorem \ref{bound3}} the hypotesis imply that $q\geq n+1$.
\end{obs}

\begin{coro}\label{mds}
Let $\mathcal{C}=(g)_{n,q}^{\theta,0}$ with $q\geq n+1$. If $\beta \in \mathbb{F}_q$, $l,c_1,\ldots ,c_r\in \mathbb{Z}_{\geq 0}$ such that 
$(c_1,\ldots ,c_r)\neq (0,\ldots ,0)$,
{\small
$$g=
\mathrm{lmcm}\ \{X-\beta ^{l+\sum _{k=1}^ri_kc_k}\ :\  
i_1=0,...,\delta,\  i_k=0,...,s_k,\ k=2,...,r,\ \delta +\sum _{k=2}^rs_k=n-k-1  \ \},$$}
and 
$N_i(\beta ^{c_j})\neq 1$ for every $i=1,...,n-1$ and $j=1,...,r$, then $\mathcal{C}$ is a Maximun Distance Separable (MDS) code.
\end{coro}

\begin{proof}
Since $\deg (g)\leq n-k$, by the Singleton bound we know that $$d_{\mathcal{C}}\leq n-\dim (\mathcal{C})+1=n-(n-\deg (g))+1\leq n-k+1.$$
On the other hand, from  Theorem \ref{bound3} it follows that $d_{\mathcal{C}}\geq n-k+1$, i.e. $d_{\mathcal{C}}=n-k+1$.
\end{proof}

\section{Magma Programs and some examples}

The following MAGMA \cite{magma} program can be used to produce MDS codes with $q\geq n+1$ when $\theta =id$:

\medskip

\begin{verbatim}
MDS:=function(q,n);
F<w>:=GF(q); R<x>:=PolynomialRing(F); C:={};
for i in [1..q-2] do
 if GCD(i,q-1) eq 1 then
 C:=C join {i};
 end if;
end for;
CC:=[c : c in C]; B:={};
for c in CC do
 for l in [0..q-2] do
  for k in [1..n-1] do
  A:=[ (w^l)*(w^c)^i : i in [0..n-k-1] ]; B:=B join {A};
  end for;
 end for;
end for;
BB:=[ b : b in B ]; D:=[ [x-i : i in BB[j] ] : j in [1..#BB] ]; 
G:={ &*D[i] : i in [1..#D] }; GG:=[g : g in G]; W1:={}; W2:={};
for s in [1..#GG] do
 h:=n-Degree(GG[s]); M:=Matrix(F,h,n,[ [Coefficient((GG[s])*x^i,j): 
 j in [0..n-1] ] : i in [0..h-1] ] ); L:=LinearCode(M); 
 d:=MinimumWeight(L); v:=Matrix(Integers(),1,4,[n,h,d,q]); g:=GG[s]; 
 print "q =", q; print "Code of type", v; print "Generator polynomial", g;
 for i in E do
  a,b:=Quotrem(x^n-i,g);
   if b eq R!0 then
    print "Constacyclic code with f =", x^n-i;
   end if;
 end for;
 W1:=W1 join {v}; W2:=W2 join {g};
end for;
return W1;
end function;
\end{verbatim}

\bigskip

The following table is made in the case
$\theta =id$ by using the command 
\begin{verbatim}
MDS(q,n);
\end{verbatim}
of the above Magma Program for $2\leq n\leq q-1\leq 6$:

\smallskip

{\tiny
\begin{longtable}{|c|c|c|c|c|c|}
\hline
 q & $n=\deg f$ & k & d & $g$ such that $g|f$ & $a$ such that $f=X^n-a$ \\
\hline
3 & 2 &  1 & 2  & $x + 2$ & 1 \\
\hline
4 & 3 &  2 & 2  & $x + w^2, x+w, x+1$ & 1\\
 &  &  1 & 3  & $x^2+x+1, x^2+w^2x+w, x^2+wx+w^2$ & 1 \\
 & 2 &  1 & 2  & $x + 1, x+w, x+w^2$ & $1, w^2, w$ \\
\hline
5 & 4 &  3 & 2  & $x + 4$ & 1 \\
 &  & 2 & 3  & $x^2+3x+1$ & $\nexists$ \\
 &  &  1 & 4  & $x^3 + 2x^2 + 3x + 4$ & $\nexists$ \\
  & 3 &  2 & 2  & $x + 4$ & 1 \\
 &  & 1 & 3  & $x^2+3x+1$ & $\nexists$ \\
  & 2 & 1 & 2  & $x+4$ & 1 \\
\hline
7 & 6 &  5 & 2  & $x + 6$ & 1 \\
 &  & 4 & 3  & $x^2+5x+1$ & $\nexists$ \\
 &  &  3 & 4  & $x^3 + 4x^2 + 3x + 6$ & $\nexists$ \\
  &  &  2 & 5  & $x^4+3x^3 + 6x^2 + 3x + 1$ & $\nexists$ \\
    &  &  1 & 6  & $x^5+2x^4+3x^3 + 4x^2 + 5x + 6$ & $\nexists$ \\
 & 5 &  4 & 2  & $x + 6$ & 1 \\ 
  &  & 3 & 3  & $x^2+5x+1$ & $\nexists$ \\
 &  &  2 & 4  & $x^3 + 4x^2 + 3x + 6$ & $\nexists$ \\
  &  &  1 & 5  & $x^4+3x^3 + 6x^2 + 3x + 1$ & $\nexists$ \\
   & 4 &  3 & 2  & $x + 6$ & 1 \\ 
  &  & 2 & 3  & $x^2+5x+1$ & $\nexists$ \\
 &  &  1 & 4  & $x^3 + 4x^2 + 3x + 6$ & $\nexists$ \\
    & 3 &  2 & 2  & $x + 6$ & 1 \\ 
  &  & 1 & 3  & $x^2+5x+1$ & $\nexists$ \\
    & 2 &  1 & 2  & $x + 6$ & 1 \\
    \hline
\caption{ \\ Example of MDS codes in $\mathbb{F}_q^n$ for $q\leq 7$ with $n\leq q-1\leq 6$}
\label{Table}
\end{longtable}}

\bigskip

The following two MAGMA \cite{magma} programs can be used 
to construct all the $(f,\Theta,0)$-skew GC codes over $\mathbb{F}_q^n$ for any polynomial $f\in
\mathbb{F}_q[X;\theta]$ given in its vectorial form in the non-commutative ($\theta\neq id$) and commutative cases ($\theta=id$) respectively:

\medskip

\begin{verbatim}
a:= ... ;
F<w>:=GF(a);

// PROGRAM 1 (non-commutative case)

R<X>:=TwistedPolynomials(F:q:= ...);

SkewGCC:=function(v);
f:=R!v; n:=Degree(f); V:=VectorSpace(F,n); W1:=[]; W2:=[]; dd:=[]; 
E:=[x : x in F | x ne 0]; S:=CartesianProduct(E,CartesianPower(F,n-1));
for ss in S do
 ll:=[ss[1]] cat [p : p in ss[2]]; q,r:=Quotrem(f,R!ll); 
  if r eq R![0] then 
   dd := dd cat [R!ll]; 
  end if;
end for;
for i in [1.. #dd] do
 if Degree(dd[i]) ge 1 then
  k:=Degree(f)-Degree(dd[i]);
  G:=Matrix(F,k,n,[V!(HorizontalJoin(Matrix(1, j+Degree(dd[i])+1, 
  Eltseq((R![0,1])^j*dd[i])), ZeroMatrix(F, 1, n-j-Degree(dd[i])-1))):
  j in {0..k-1}]); L:=LinearCode(G); dd[i]; print " "; G; print " "; 
  print "Code of type: ", n, k, MinimumWeight(L);
  print "-------------"; W1:= W1 cat [k]; W2:= W2 cat [MinimumWeight(L)];
 end if;
end for;
print "Spectrum of the distances for", f; 
print "n=", n; print "k="; W1; print "d="; 
return W2; 
end function;

// PROGRAM 2 (commutative case)

R<x> := PolynomialRing(F);

GCC:=function(v); 
f:=R!v; n:=Degree(f); W1:=[]; W2:=[];
for i in [1..#Factorisation(f)] do
 if Factorisation(f)[i][2] eq 1 then
  a:=R!Factorisation(f)[i][1]; k:=Degree(f)-Degree(R!a); 
  G:=Matrix(F,k,n,[[Coefficient((R!a)*x^i,j): j in {0..n-1}] 
  : i in {0..k-1}]); L:=LinearCode(G); a; print " "; G; print " "; 
  print "Code of type: ", n, k, MinimumWeight(L); 
  print "-------------"; W1:= W1 cat [k]; W2:= W2 cat [MinimumWeight(L)];
 end if;
 if Factorisation(f)[i][2] ne 1 then
  for j in [1.. Factorisation(f)[i][2]] do
   a:=(R!Factorisation(f)[i][1])^j; k:=Degree(f)-Degree(R!a); 
   G:=Matrix(F,k,n,[[Coefficient((R!a)*x^i,j): j in {0..n-1}] 
   : i in {0..k-1}]); L:=LinearCode(G); a; print " "; G; print " "; 
   print "Code of type: ", n, k, MinimumWeight(L); 
   print "-------------"; W1:= W1 cat [k]; W2:= W2 cat [MinimumWeight(L)];
  end for;
 end if;
end for;
print "Spectrum of the distances for", f; 
print "n=", n; print "k="; W1; print "d="; 
return W2; 
end function;
\end{verbatim}

\begin{ejem}
Consider $f=X^4 + X^3 + w X^2 + 1\in\mathbb{F}_8[X,\theta]$. If $\theta (z)=z^2$ for any $z\in\mathbb{F}_8$, then
by {\em PROGRAM 1} (with $v=[1,0,w,1,1], a=8$ and $q:=2$) we obtain $28$ MDS skew GC codes with parameters $[4,1,4]_8, [4,2,3]_8$ and $[4,3,2]_8$, while 
if $\theta = id$, then
by {\em PROGRAM 2} (with $v=[1,0,w,1,1], a=8$) we get only $2$ MDS GC codes with parameters $[4,1,4]_8$ and $[4,3,2]_8$.
Moreover, the MDS skew GC code of type $[4,2,3]_8$ is given, for instance, by the right divisor $X^2 + wX + w$ of $f$.
\end{ejem}

\bigskip

\begin{ejem}
Consider the finite field $\mathbb{F}_4=\mathbb{F}_2[\alpha]$ with $\alpha ^2+\alpha+1=0$, $\theta$ the Frobenius automorphism and $\beta\in\mathbb{F}_4$. 
Let $f=X^8+X^6+X^2+1\in N(R)$ with $R=\mathbb{F}_4[X,\theta,\delta_{\beta}^{\theta}]$ and note that $f=f_1^{\alpha_1}\cdot f_2^{\alpha_2}$
is a factorization of $f$ as in $(\#)$, where
$(f_1,\alpha_1)=(X^2+1,2)$ and $(f_2,\alpha_2)=(X^4+X^2+1,1)$. Set
$U_{i}:=\textrm{Ker} f_{i}^{\alpha _i}(T_f)$ for $i=1,2$, where $T_f$ is given by the following rule:
$$T_f(v_1,...,v_8):=(v_1^2,...,v_8^2)\left(
\begin{array}{c|ccccccc}
0&1& & & & & & \\
0 & & 1 & & & & & \\
0 & & & 1 & & & & \\
0 & & & & 1 & & & \\
0 & & & &  & 1 & & \\
0 & & & & & & 1 &  \\
0 & & & & & & & 1\\
\hline
1& 0 & 1 & 0 & 0 & 0 & 1 & 0
\end{array}
\right)
+\beta(v_1^2-v_1,...,v_8^2-v_8),$$
i.e. $T_f(v_1,...,v_8):=(v_8^2,v_1^2,v_2^2+v_8^2,v_3^2,v_4^2,v_5^2,v_6^2+v_8^2,v_7^2)+\beta(v_1^2-v_1,...,v_8^2-v_8)$.

\bigskip

By the {\em MAGMA} program 

\begin{verbatim}
F<w>:=GF(4); P<[x]>:=PolynomialRing(F,8); A:=KSpace(F,8); B:=[a : a in A];
for b in F do
T:= map < A -> A | v :-> [v[8]^2+b*(v[1]^2+v[1]), v[1]^2+b*(v[2]^2+v[2]), 
v[8]^2+v[2]^2+b*(v[3]^2+v[3]), v[3]^2+b*(v[4]^2+v[4]), 
v[4]^2+b*(v[5]^2+v[5]), v[5]^2+b*(v[6]^2+v[6]), 
v[6]^2+v[8]^2+b*(v[7]^2+v[7]), v[7]^2+b*(v[8]^2+v[8])] >;
U1:={}; U2:={}; g:=T^4; h:=T^2; f1:= map < A -> A | x :-> g(x)+x >;
f2:= map < A -> A | x :-> g(x)+h(x)+x >; o:=A![0,0,0,0,0,0,0,0]; 
for a in B do
if f1(a) eq o then
U1:=U1 join {a};
end if;
if f2(a) eq o then
U2:=U2 join {a};
end if;
end for;
UU1:=[u : u in U1]; UU2:=[u : u in U2]; V1:={}; V2:={};
for i in [4..#UU1] do
for j in [4..#UU2] do
G1:= sub< A | UU1[1], UU1[2],UU1[i-1],UU1[i] >;
G2:= sub< A | UU2[1], UU2[2],UU2[j-1],UU2[j] >;
if Dimension(G1) eq 4 then
V1:= V1 join {Basis(G1)};
end if;
if Dimension(G2) eq 4 then
V2:= V2 join {Basis(G2)};
end if;
end for;
end for;
VV1:=[v : v in V1]; VV2:=[w : w in V2]; b; VV1[1]; VV2[1]; print "---";
end for;
\end{verbatim}

\bigskip

we obtain the following table: 

{\em 

{\tiny
\begin{longtable}{|c|c|c|}
\hline
 $\beta$ & Generator matrix of $U_1$ & Generator matrix of $U_2$ \\
\hline
 & & \\
$1$ & 
$\left( \begin{array}{cccccccc} 
1&0 &0 &0 &0 &0 &1 &0 \\
0&1 &0 &0 &0 &0 &0 &1 \\
0& 0&1 &0 &1 &0 &1 &0 \\
0& 0& 0& 1& 0& 0&0 & 1
\end{array}\right)$

& 

$\left(\begin{array}{cccccccc} 
1& 0& 0& 0& 1& 0& 0& 0\\
0& 1& 0& 0& 0& 1& 0& 0\\
0& 0& 1& 0& 0& 0& 1& 0\\
0& 0& 0& 1& 0& 0& 0& 1
\end{array}\right)$ \\
& & \\
\hline
& & \\
$0,w,w^2$ & 

$\left(\begin{array}{cccccccc} 
1& 0& 0& 0& 0& 0& 1& 0\\
0& 1& 0& 0& 0& 0& 0& 1\\
0& 0& 1& 0& 1& 0& 1& 0\\
0& 0& 0& 1& 0& 1& 0& 1
\end{array}\right)$

& 

$\left(\begin{array}{cccccccc} 
1& 0& 0& 0& 1& 0& 0& 0\\
0& 1& 0& 0& 0& 1& 0& 0\\
0& 0& 1& 0& 0& 0& 1& 0\\
0& 0& 0& 1& 0& 0& 0& 1
\end{array}\right)$ \\
& & \\
\hline
\caption{ \\ Example of skew GC codes in $\mathbb{F}_4^8$ with a non-trivial derivation}
\label{Table-bis}
\end{longtable}}}

\end{ejem} 

\medskip

\section*{Conclusion}

\noindent In this paper, we consider codes invariant by a pseudo-linear transformation of $\mathbb{F}_q^n$ for $n\geq 2$, 
called skew generalized cyclic (GC) codes, where $\mathbb{F}_q$ is a finite field with $q$ elements. 
We study some of their main algebraic and geometric properties and when the derivation is trivial, 
we find the minimal polynomial of a pseudo-linear transformation and we give some lower bounds for the minimum Hamming distance of a skew GC code. 
Finally, examples and Magma programs are given as applications of some theoretical results.

\end{document}